\documentclass[final,12pt]{article}
\usepackage{hyperref}
\usepackage{amssymb,amsmath,amsthm}
\usepackage{cite}
\usepackage{tikz}
\usetikzlibrary{arrows,backgrounds,decorations.pathmorphing,decorations.pathreplacing,positioning,fit}
\usepackage{enumerate}
\usepackage[conditional,light,first,bottomafter]{draftcopy}
\draftcopyName{DRAFT\space\today}{130}
\draftcopySetScale{65}
\usepackage[letterpaper,hmargin=3cm,vmargin=3cm]{geometry}
\usepackage{setspace}
\makeatletter
\renewcommand{\section}{\@startsection%
{section}%
{1}%
{0em}%
{1.7em}%
{1.2em}%
{\normalfont\large\centering\bfseries}}
\renewcommand{\@seccntformat}[1]%
{\csname the#1\endcsname.\hspace{0.5em}}
\makeatother

\numberwithin{equation}{section}
\newtheorem{theorem}{Theorem}[section]
\newtheorem{proposition}[theorem]{Proposition}
\newtheorem{lemma}[theorem]{Lemma}
\newtheorem{corollary}[theorem]{Corollary}
\theoremstyle{definition}
\newtheorem{definition}[theorem]{Definition}
\newtheorem{remark}[theorem]{Remark}
\newtheorem*{notation}{Notation}

\newcommand{\abs}[1]{\left|#1\right|}
\newcommand{\norm}[1]{\left\|#1\right\|}
\newcommand{\inner}[2]{\left\langle#1,#2\right\rangle}
\newcommand{\cc}[1]{\overline{#1}}
\newcommand{\reals}{\mathbb{R}}
\newcommand{\nats}{\mathbb{N}}
\newcommand{\complex}{\mathbb{C}}
\newcommand{\eval}[1]{\upharpoonright_{#1}}
\DeclareMathOperator{\im}{Im}
\DeclareMathOperator{\dom}{dom}
\DeclareMathOperator{\ran}{ran}
\DeclareMathOperator{\rank}{rank}

\DeclareMathOperator{\Span}{span}
\DeclareMathOperator{\card}{\#}
\DeclareMathOperator{\supp}{supp}
\DeclareMathOperator{\ind}{ind}
\begin{document}
\begin{titlepage}
\title{Stability of determinacy and inverse spectral problems for
  Jacobi operators
\footnotetext{%
Mathematics Subject Classification(2010):
41A10  
47B36, 
33E30  
}
\footnotetext{%
Keywords:
Index of determinacy;
Density of polynomials;
Green function;
Jacobi operators.
}
\hspace{-8mm}
\thanks{%
Research partially supported by UNAM-DGAPA-PAPIIT IN105414
}%
}
\author{
\textbf{Rafael del Rio}
\\
\small Departamento de F\'{i}sica Matem\'{a}tica\\[-1.6mm]
\small Instituto de Investigaciones en Matem\'aticas Aplicadas y en Sistemas\\[-1.6mm]
\small Universidad Nacional Aut\'onoma de M\'exico\\[-1.6mm]
\small C.P. 04510, M\'exico D.F.\\[-1.6mm]
\small \texttt{delrio@iimas.unam.mx}
\\[2mm]
\textbf{Luis O. Silva\thanks{Parcially supported by SEP-CONACYT 254062}}
\\
\small Departamento de F\'{i}sica Matem\'{a}tica\\[-1.6mm]
\small Instituto de Investigaciones en Matem\'aticas Aplicadas y en Sistemas\\[-1.6mm]
\small Universidad Nacional Aut\'onoma de M\'exico\\[-1.6mm]
\small C.P. 04510, M\'exico D.F.\\[-1.6mm]
\small \texttt{silva@iimas.unam.mx}
}
\date{}
\maketitle
\vspace{-4mm}
\begin{center}
\begin{minipage}{5in}
  \centerline{{\bf Abstract}} \bigskip
  This work studies the interplay between Green functions, the index
  of determinacy of spectral measures and interior finite rank
  perturbations of Jacobi operators. The index of determinacy
  quantifies the stability of uniqueness of solutions of the moment
  problem. We give results on the constancy of this index in terms of
  perturbations of the corresponding Jacobi operators. The permanence
  of the $N$-extremality of a measure is also studied. A measure $\mu$
  is $N$-extremal when the polynomials are dense in
  $L_2(\reals,\mu)$. As a by-product, we give a characterization of
  the index in terms of cyclic vectors. We consider a new inverse
  problem for Jacobi operators in which information on the place where
  the interior perturbation occurs is obtained from the index of
  determinacy.
\end{minipage}
\end{center}
\thispagestyle{empty}
\end{titlepage}
\section{Introduction}
\label{sec:intro}

Given a sequence $\{s_k\}_{k=0}^\infty$ of real numbers, the problem
of finding a Borel measure $\mu$ in $\reals$ such that
\begin{equation*}
  s_k=\int_\reals t^kd\mu\quad \text{ for all } k=0,1,2,\dots
\end{equation*}
is called the Hamburger moment problem.

Denote by $\mathcal{M}$ the set of Borel measures on $\reals$ with
infinite support and all their moments finite. For a positive sequence
$\{s_k\}_{k=0}^\infty$ (see definition in \cite[Chap.\,1,
Sec\,1]{MR0184042}), the corresponding Hamburger moment problem has
always a solution $\mu\in\mathcal{M}$
\cite[Thm.\,2.1.1]{MR0184042}. $\mathcal{M}$ splits into two sets, one
corresponding to the case when the Hamburger moment problem has a
unique solution and the other when it has various solutions.  In the
first case, it is said that the moment problem is determinate,
whereas, in the second case, the problem is called indeterminate. If a
moment problem is determinate (indeterminate), then the corresponding
solution, that is the measure, is also called determinate
(indeterminate).

The problem of finding conditions under which a relevant class of
functions is dense in the spaces $L_p(\reals, \mu)$ is classical in
analysis. In particular, conditions which guarantee density of
polynomials go back at least to the work of Hamburger
\cite{MR1511981}. For related work see for example \cite{MR0084064,
  MR1217081, MR2562213, MR2172989, MR1842874, MR2071709, MR2214480,
  MR2415042, MR638619, MR1871391, MR1650247} (see in
\cite[Sec.\,4.8]{MR1864396} a brief compilation of results on the
matter). A fundamental result characterizing the measures
$\mu\in\mathcal{M}$ for which the polynomials are dense in
$L_2(\reals,\mu)$ is due to M. Riesz \cite[Thm.\,2.3.2]{MR0184042},
\cite{mriesz1923}. It establishes that for the polynomials to be dense
in $L_2(\reals,\mu)$ it is necessary and sufficient that $\mu$ be
$N$-extremal (see definition in \cite[Pag.\,43]{MR0184042}). In
\cite[Pag.\,86]{MR1627806}, $N$-extremal solutions are called von
Neumann solutions whereas in \cite[Pag.\,2796]{MR1308001} $N$-extremal means
Nevanlinna extremal. Note that in contrast to the definition given in
\cite{MR1308001}, here all determinate solutions are $N$-extremal
\cite[Cor.\,2.3.3]{MR0184042}.

A concept related to the determinacy and $N$-extremality of a measure
is the concept of the \emph{index of determinacy} introduced by Berg
and Dur\'{a}n in \cite{MR1308001}. The index of determinacy of
$\mu\in\mathcal{M}$ quantifies the stability of $\mu$ to be the unique
solution of a moment problem under perturbations of it (see
Definition~\ref{def:index}). This index also gives information on how
a measure can be perturbed and maintain the property of being
$N$-extremal. The fact that a measure $\mu\in\mathcal{M}$ is
indeterminate $N$-extremal or determinate may be changed by adding or
substracting the mass at only one point (see
Proposition~\ref{prop:duran} below) or by modifying the weights
without changing the support (see Proposition~\ref{prop:changing-weights}).

Jacobi operators, i.\,e., self-adjoint extensions of operators having
a tridiagonal matrix representation (see \eqref{eq:jm-0}), naturally
appear in the theory of the Hamburger moment problem. It turns out
that every $N$-extremal solution of a Hamburger moment problem
normalized so that $s_0=1$ is the spectral measure of a Jacobi
operator (see Theorem~\ref{thm:nec-suf-cond-for-meausure}). Thus, the
study of measures $\mu$ such that the polynomials are dense in
$L_2(\reals,\mu)$ is the study of self-adjoint extensions of operators
having a semi-infinite Jacobi matrix as its matrix representation.

We study the constancy of the index of determinacy after changing the
weights and support of the measure. The permanence of the
$N$-extremality of a given measure is also considered. Similar
questions on stability are studied in \cite{MR3407921} and
\cite{MR1754999}.  Our approach to this matter is mainly based on
Jacobi operators and Green functions.  This allows us to give
results on the stability of the index of determinacy of the spectral
measure of a Jacobi operator under finite rank perturbations of the
operator. Jacobi operators model linear mass-spring systems and the
perturbations considered here correspond to changing one mass and
spring constant in some place of the chain.

Our findings on the stability of the index of determinacy and the
$N$-ex\-trem\-ality of the spectral measures of Jacobi operators shed
light on the inverse spectral problem of reconstructing an operator
from its spectrum and the spectrum of an interior finite rank
perturbation of it. It turns out that the aforehand knowledge of the
index of determinacy of the spectral measure of the Jacobi operator
determines the place where the interior perturbation
occurs. Remarkably, for finite Jacobi matrices, as well as in the case
of infinite index of determinacy, one cannot recover the place of the
perturbation.

As a by-product of our research, we give a new characterization of the
index of determinacy in terms of the cyclicity of vectors generated by
polynomials functions of Jacobi operators.

The paper is organized as follows. In the next section we give some
preparatory facts on Jacobi operators. In
Section~\ref{sec:green-functions}, the Weyl and Green functions
associated to Jacobi operators are introduced and we prove a criterion
for a Green function to be a Weyl function
(Theorem~\ref{thm:equivalence-green-weyl}). This result is interpreted
later in terms of the index of determinacy
(Corollary~\ref{cor:index-weyl-function}). Section
\ref{sec:stability-index-determinacy} presents a characterization of
the index of determinacy (Corollary~\ref{cor:other-criterion}) and
establishes stability results for the index. We provide conditions for
two measures with the same support and different weights to have the
same index (Theorems~\ref{thm:main-minimal-equal-indices} and
\ref{thm:stability-determ-index-weights}).  Moreover, conditions for
two measures with different supports to have the same index are found
(Corollary~\ref{cor:finite-rank-pert-index-determinacy}). We show that
finite-rank perturbations of Jacobi operators do not modify the index
of determinacy of the corresponding
measures. Section~\ref{subsec:sol-inv-problem} presents a new
development in the inverse spectral analysis of interior perturbations
of Jacobi operators. We consider a  two-spectra inverse problem
where the information of the index of determinacy is given in
advance. This section connects the results of previous sections to the
inverse spectral problem studied in \cite{MR3634444}.  To the
best of our knowledge, this is the first time that the index of
determinacy is used in inverse spectral theory.

\section{Jacobi operators}
\label{sec:jacobi_operators}

For a sequence $f=\{f_k\}_{k=1}^\infty$ of complex numbers, consider the
second order difference expressions
\begin{subequations}
  \label{eq:difference-expr}
\begin{align}
   \label{eq:difference-recurrence}
  (\Upsilon f)_k&:= b_{k-1}f_{k-1} + q_k f_k + b_kf_{k+1}
  \quad k \in \mathbb{N} \setminus \{1\},\\
   \label{eq:difference-initial}
   (\Upsilon f)_1&:=q_1 f_1 + b_1 f_2\,,
\end{align}
\end{subequations}
where $q_k\in\mathbb{R}$ and $b_k>0$ for any $k\in\mathbb{N}=\{1,2,\dots\}$. We
remark that (\ref{eq:difference-initial}) can be seen as a boundary
condition.

\begin{definition}
  Let $l_2(\mathbb{N})$ be the space of square summable complex
  sequences. In this Hilbert space, define the operator $J_0$ whose
  domain is the set of sequences having a finite number of
  non-zero elements and is given by $J_0f:=\Upsilon
  f$.\label{def:j-nought}
\end{definition}
Clearly, the operator
$J_0$ is symmetric and therefore closable, so one can consider the
operator $\cc{J_0}$ being its closure. By the definition of the matrix
representation of an unbounded symmetric operator given in \cite[Sec. 47]{MR1255973},
$\overline{J_0}$ is the operator whose matrix representation with
respect to the canonical basis $\{e_n\}_{n=1}^\infty$ in
$l_2(\mathbb{N})$ is
\begin{equation}
  \label{eq:jm-0}
  \begin{pmatrix}
    q_1 & b_1 & 0  &  0  &  \cdots
\\[1mm] b_1 & q_2 & b_2 & 0 & \cdots \\[1mm]  0  &  b_2  & q_3  &
b_3 &  \\
0 & 0 & b_3 & q_4 & \ddots\\ \vdots & \vdots &  & \ddots
& \ddots
  \end{pmatrix}\,.
\end{equation}
Recall that the element $e_n$ of the canonical basis is the sequence
whose elements are zero except for the $n$-th entry which is 1.  Thus
$\cc{J_0}$ is the minimal closed symmetric operator such that
$\inner{e_j}{\cc{J_0}e_k}$ is the $j,k$ entry of the matrix above.

\begin{remark}
\label{rem:deficiency-indices}
The deficiency indices of the symmetric operator $\cc{J_0}$ are either
$(1,1)$ or $(0,0)$ \cite[Chap.\,4,\,Sec.\,1.2]{MR0184042},
\cite[Chap.\,7 Thm.\,1.1]{MR0222718}. When
$\cc{J_0}$ has deficiency indices $(1,1)$, respectively $(0,0)$, the
matrix \eqref{eq:jm-0} is said to be in the limit circle case,
respectively limit point case \cite[Def.\,1.3.2]{MR0184042}. Thus, if $J$ is a self-adjoint
extension of $J_0$,then  either $J$ is a proper closed
symmetric extension of $\cc{J_0}$ or $J=\overline{J_0}$.
\end{remark}
\begin{definition}
 \label{def:jacobi-operator}
 Given the matrix \eqref{eq:jm-0},
we consider $J$ a fixed self-adjoint extension of $J_0$ and refer to it as the
Jacobi operator associated with \eqref{eq:jm-0}.
\end{definition}
When \eqref{eq:jm-0}
is in the limit circle case, there are more than one Jacobi operators
associated with the matrix \eqref{eq:jm-0}.

By setting $f_1=1$, a solution of the equations
\begin{subequations}
\label{eq:spectral-equation}
\begin{align}
  \label{eq:initial-spectral}
  (\Upsilon f)_1&:= zf_1\,,\\
  \label{eq:recurrence-spectral}
  (\Upsilon f)_k&:= z f_k\,,
  \quad k \in \mathbb{N} \setminus \{1\},
\end{align}
\end{subequations}
can be found uniquely by recurrence. This solution, denoted by
\begin{equation}
\label{eq:pi}
\pi(z)=\{\pi_k(z)\}_{k=1}^\infty \,,
\end{equation}
is such that $\pi_k(z)$ is a polynomial of
degree $k-1$.
The elements of the sequence $\pi(z)$
are referred to as the polynomials of the first
kind associated to the matrix
(\ref{eq:jm-0}). By comparing (\ref{eq:difference-expr}) with
(\ref{eq:spectral-equation}), one concludes that for $\pi(z)$ to be in
$\ker(J_0^*-zI)$, it is necessary and sufficient that $\pi(z)$ be an
element of $l_2(\nats)$. Of course, $\pi(z)\in\ker(J-zI)$, if and only
if $\pi(z)\in\dom(J)$.

Observe that
\begin{align*}
  J e_k&= b_{k-1}e_{k-1} + q_k e_k + b_ke_{k+1}
  \quad k \in \mathbb{N} \setminus \{1\},\\
   J e_1&=q_1 e_1 + b_1 e_2\,,
\end{align*}
Thus, by the definition of $\pi_k(z)$, one has
\begin{equation}
  \label{eq:delta-k-through-delta-1}
  e_k=\pi_k(J)e_1\quad\forall k\in\nats\,.
\end{equation}
This implies that $J$ is simple and $e_1$ is a cyclic vector (see
\cite[Sec. 69]{MR1255973}). Therefore, if one defines
\begin{equation}
  \label{eq:spectral-function-def}
  \rho(t):=\inner{e_1}{E_J(t)e_1}\,,\quad t\in\reals\,,
\end{equation}
where $E_J$ is the resolution of the identity given by the spectral
theorem, then, by \cite[Sec. 69, Thm. 2]{MR1255973}), there is a
unitary map $\Phi:L_2(\reals,\rho)\to l_2(\nats)$ such that
$\Phi^{-1}J\Phi$ is the multiplication by the independent variable
defined in its maximal domain. Henceforth we identify the function
$\rho(t)$ with the Borel measure $\rho$ which it uniquely determines
and call it spectral measure of $J$ (with respect to $e_1$).
Moreover, due to \cite[Sec. 69,
Thm. 2]{MR1255973}), it follows from
(\ref{eq:delta-k-through-delta-1}) that
the function $\pi_k\eval{\reals}$ belongs to $L_2(\reals,\rho)$ for
all $k\in\nats$, i.\,e., all moments of $\rho$ are finite (see also
\cite[Thm.\,4.1.3]{MR0184042}). The equation
\eqref{eq:delta-k-through-delta-1} means that
\begin{equation}
  \label{eq:unitary-map}
  \Phi\pi_k=e_k\,,\quad\forall k\in\nats\,,
\end{equation}
which implies that the polynomials are dense in $L_2(\reals,\rho)$
since $\Phi$ is unitary. Note also that, due to
\eqref{eq:spectral-function-def}, $\int_\reals d\rho=1$ holds.

Now, assume that one is given a measure $\rho$ satisfying
$\int_\reals d\rho=1$ and such that all the polynomials are in
$L_2(\reals,\rho)$ and they are dense in this space. Consider the
operator of multiplication by the independent variable $A$ in
$L_2(\reals,\rho)$ with
\begin{equation*}
  \dom(A)=\{f\in L_2(\reals,\rho):\int_\reals t^2\abs{f}^2d\rho(t)<+\infty\}
\end{equation*}
This operator is self-adjoint and $E_A(\Delta)=\chi_\Delta$, where $\chi_\Delta$
is the characteristic function of the Borel set
$\Delta\subset\reals$. Therefore, similar to \eqref{eq:spectral-function-def},
\begin{equation*}
  \rho(\Delta)=\inner{1}{\chi_\Delta 1}
\end{equation*}
for any Borel set
$\Delta\subset\reals$. Note that $1$ is a cyclic vector for $A$ since
the polynomials are dense in $L_2(\reals,\rho)$. Applying the
Gram-Schmidt procedure to the sequence $\{1,t,t^2,\dots\}$, one obtains an
orthonormal basis $\{p_1=1,p_2,p_3\dots\}$ contained in the domain of
$A$. It can be verified \cite[Sec.\,69]{MR1255973}
(cf. \cite[Pags.\,92,93]{MR1627806}) that the matrix
\begin{equation}
  \label{eq:matrix-multiplication-operator}
  a_{jk}=\inner{p_j}{Ap_k}\quad\forall j,k\in\nats\,.
\end{equation}
is a Jacobi matrix. According to \cite[Sec.\,47]{MR1255973}, $A$ is a
self-adjoint extension of an operator whose matrix representation is
\eqref{eq:matrix-multiplication-operator}.

By constructing an isometry $\Psi$ between $L_2(\reals,\rho)$ and
$l_2(\nats)$ such that $\Psi p_k=e_k$, one arrives at the following
central assertion (cf. \cite[Thms.\,2.3.3 and 4.1.4]{MR0184042}).
\begin{theorem}
  \label{thm:nec-suf-cond-for-meausure}
 A measure $\rho$ is the spectral measure of a Jacobi operator if
 and only if $\int_\reals d\rho=1$, all the polynomials are in
  $L_2(\reals,\rho)$ and they are dense in this space.
\end{theorem}

\begin{remark}
  \label{rem:finite-matrix}
  Any probability measure with finite support is the spectral measure of the
  operator associated with some finite Jacobi matrix.
\end{remark}

\begin{definition}
  \label{def:weyl-function}
The Weyl $m$-function is defined as follows
\begin{equation}
  \label{eq:weyl-function}
  m(z):=\inner{e_1}{(J-z I)^{-1}e_1}\,,\qquad z\not\in\sigma(J)\,.
\end{equation}
Here, for a given operator $T$, $\sigma(T)$ denotes the spectrum of
it.
\end{definition}
Using the map $\Phi$, one concludes from this
definition that
\begin{equation}
  \label{eq:weyl-by-spectral-th}
  m(z)=\int_{\mathbb{R}}
  \frac{d\rho(t)}{t-z}\,.
\end{equation}
Thus, by the Nevanlinna representation theorem (see
\cite[Thm.\,5.3]{MR1307384}), $m(z)$ is a Herglotz function. Recall
that a function $f$ is Herglotz (also called Pick or Nevanlinna-Pick
function) when it is holomorphic in the upper half-plane and $\im
f(z)\ge 0$ whenever $\im z>0$.

\section{Green functions for Jacobi operators}
\label{sec:green-functions}

\begin{definition}
\label{def:green-function}
We use the following notation
\begin{equation*}
  G(z,k):=\inner{e_k}{(J-zI)^{-1}e_k}\qquad z\not\in\sigma(J)
\end{equation*}
and call $G(z,k)$ the $k$-th Green function of the Jacobi operator $J$.
Observe that $G(z,1)=m(z)$ (See Definition~\ref{def:weyl-function}).
\end{definition}
In view of
(\ref{eq:delta-k-through-delta-1}) and \eqref{eq:spectral-function-def}, one has
\begin{equation}
  \label{eq:green-integral}
  G(z,n)=\int_\reals\frac{\pi_n^2(t)d\rho(t)}{t-z}\,.
\end{equation}
Thus, for any $n\in\nats$, $G(\cdot,n)$ is a Herglotz function. This
function is extended analytically to the eigenvalues of $J$ which are simultaneously zeros
of $\pi_n$ since these points are removable singularities.

Using the
von Neumann expansion for the resolvent
(cf.\cite[Chap.\,6,\,Sec.\,6.1]{MR1711536})
\begin{equation*}
  (J-zI)^{-1}e_n=
  -\sum_{k=0}^{N-1}\frac{J^ke_n}{z^{k+1}}
  +\frac{J^N}{z^{N}}
  (J-z I)^{-1}e_n\,,\quad N\in\nats\,,
\end{equation*}
where $z\in\mathbb{C}\setminus\sigma(J)$,
one obtains the following asymptotic formula
\begin{equation}
  \label{eq:G=-asymptotics}
  G(z,n)=-\frac{1}{z} +O(z^{-2})
\end{equation}
as $z\to\infty$ along any ray intersecting the real axis only at $0$.

The following definition is taken from \cite[Def.\,2.1]{MR3377115}.
\begin{definition}
  \label{def:submatrices}
For a subspace $\mathcal{G}\subset\l_2(\nats)$ (therefore
$\mathcal{G}$ is closed), let
  $P_{\mathcal{G}}$ be the orthogonal projection onto
  $\mathcal{G}$. Also, define
  $\mathcal{G}^\perp:=\{\phi\in
  l_2(\nats):\inner{\phi}{\psi}=0\,\forall \psi\in\mathcal{G}\}$ and
  the subspace
$\mathcal{F}_n:=\Span\{e_k\}_{k=1}^n$.
  For the operator $J$ given in Definition~\ref{def:jacobi-operator}, consider the
  operators
  \begin{equation}
    \label{eq:j+-}
    J_n^+:=P_{\mathcal{F}_n^\perp}J\eval{\mathcal{F}_n^\perp}\quad n\in\nats\,,\quad
    J_n^-:=P_{\mathcal{F}_{n-1}}J\eval{\mathcal{F}_{n-1}}\quad n\in\nats\setminus\{1\}\,.
  \end{equation}
 Here, we have used the notation
 $J\eval{\mathcal{G}}$ for the restriction of $J$ to the set
 $\mathcal{G}$, that is,
 $\dom(J\eval{\mathcal{G}})=\dom(J)\cap\mathcal{G}$. Consider also the corresponding $m$-Weyl functions
 \begin{equation}
   \label{eq:def-m-functions}
   m_n^+(z):=\inner{e_{n+1}}{(J_n^+-zI)^{-1}e_{n+1}}\,,\qquad
    m_n^-(z):=\inner{e_{n-1}}{(J_n^--zI)^{-1}e_{n-1}}\,.
 \end{equation}
\end{definition}
\begin{remark}
\label{rem:same-boundary}
The operator $J_n^+$ is a self-adjoint extension of the operator whose
matrix representation with respect to the basis
$\{e_k\}_{k=n+1}^\infty$ of the Hilbert space
$(\Span\{e_k\}_{k=1}^n)^\perp$ is (\ref{eq:jm-0}) with
the first $n$ rows and $n$ columns removed. When $J_0$ is not
essentially self-adjoint, $J_n^+$ has the same boundary conditions at
infinity as the operator $J$. Clearly, the operator $J_n^-$ lives in
an $n-1$-dimensional space.
\end{remark}

\begin{remark}
  \label{rem:zeros-poly}
  By \cite[Cor.\,2.3]{MR3377115}, the set of zeros of the polynomial $\pi_n$
  coincides with the spectrum of $J_n^-$.
\end{remark}

\begin{remark}
  \label{rem:intersections-spectra}
  It follows from \cite[Lem.\,2.9, and Prop.\,3.3]{MR3377115} that
  \begin{equation*}
    \sigma(J_n^-)\cap\sigma(J_n^+)=\sigma(J_n^-)\cap\sigma(J)\,.
  \end{equation*}
\end{remark}

The next assertion is proven in \cite[Thm. 2.8]{MR1616422}  and \cite[Prop.,2.3]{MR3377115}.
\begin{proposition}
  \label{prop:Gkk-formula}
  For any $n\in\nats$
\begin{equation}
  \label{eq:Gkk-formula2}
   G(z,n)=\frac{-1}{b_n^2m_n^+(z)+b_{n-1}^2m_n^-(z)+z-q_n}\,,
\end{equation}
where we define $m^-_1(z)\equiv 0$.
\end{proposition}
\begin{notation}
  Let us denote by $\mu_n$ and $\sigma_n$ the measures given by the
  Nevanlinna representation of the function $m_n^-(z)$ and $m_n^+(z)$,
  respectively, that is, $m_n^\pm$ given in \eqref{eq:def-m-functions}
  are the Borel transforms of $\mu_n$ and $\sigma_n$. Also, denote by
  $\rho_n$ the measure given by the Nevanlinna representation of the
  function $G(z,n)$. Thus
 \begin{align}
\label{eq:def-sigma-n}
   m_n^+(z)&=\int_\reals\frac{d\sigma_n(t)}{t-z}\\
\label{eq:def-mu-n}
 m_n^-(z)&=\int_\reals\frac{d\mu_n(t)}{t-z}\\
\label{eq:def-rho-n}
 G(z,n)&=\int_\reals\frac{d\rho_n(t)}{t-z}\,.
 \end{align}
We denote by  $\delta_\lambda$ the measure
\begin{equation}
  \label{eq:dirac-measure}
  \delta_\lambda(\Delta):=
  \begin{cases}
    1 & \lambda\in\Delta\\
    0 & \lambda\not\in\Delta
  \end{cases}
  \end {equation}
where  $\Delta\subset\reals $ is a Borel set.
  \end{notation}
\begin{theorem}
  \label{thm:equivalence-green-weyl}
  Fix $n\in\nats$ and let $G(z,n)$ be the $n$-th Green function of the
  Jacobi operator $J$. If the polynomials are dense
  in $L_2(\reals,\rho_n)$, then $G(z,n)$ is the $l$-th Green function
  of some other Jacobi operator for any $l\in\nats$.
\end{theorem}
\begin{proof}
  We show that
  the measure $\rho_n$ satisfies the conditions of
  Theorem~\ref{thm:nec-suf-cond-for-meausure}. In view of
  \eqref{eq:green-integral} and \eqref{eq:def-rho-n}, for any $n\in\nats$,
  \begin{equation*}
    \int_\reals d\rho_n= \int_\reals\pi_n^2(t) d\rho=\norm{\pi_n(\cdot)}_{L_2(\reals,\rho)}^2=1\,,
  \end{equation*}
where the last equality holds due to \eqref{eq:unitary-map}. Moreover,
for any $m\in\nats\cup\{0\}$,
\begin{equation*}
   \int_\reals t^md\rho_n(t)= \int_\reals t^m\pi_n^2(t) d\rho(t)<\infty
\end{equation*}
since all the moments of $\rho$ are finite. Thus all the polynomials
are in $L_2(\reals,\rho_n)$ and by hypothesis the polynomials are
dense there.  Therefore Theorem~\ref{thm:nec-suf-cond-for-meausure},
taking into account \eqref{eq:def-rho-n} and \eqref{eq:weyl-by-spectral-th},
implies that $G(z,n)$ is the Weyl $m$-function of some Jacobi
operator.

Let $m(z)$ be the Weyl $m$-function of some Jacobi operator $J$. We show that $m(z)$ is the
$l$-th Green function for any $l\in\nats$. By
Proposition~\ref{prop:Gkk-formula} one has
\begin{equation*}
  -m(z)^{-1}=b_1^2m_1^+(z)+z-q_1=z-q_1+\sum_{k=1}^\infty\frac{\eta_k}{\alpha_k-z}\,.
\end{equation*}
Thus, since $m_1^+$ is the Weyl $m$-function of the Jacobi operator $J_1^+$, it
follows from Theorem~\ref{thm:nec-suf-cond-for-meausure} that the measure
\begin{equation*}
  \sigma:=\sum_{k=1}^\infty \eta_k \delta_{\alpha_k}
\end{equation*}
is such that the polynomials are in $L_2(\reals,\sigma)$ and they are dense
in this space. One can also write
\begin{equation}
  \label{eq:almost-green}
   -m(z)^{-1}=z-q_1+\left(\sum_{k=1}^{l-1}+\sum_{k=l}^\infty\right)\frac{\eta_k}{\alpha_k-z}\,.
\end{equation}
Note that the measure
\begin{equation*}
  \widetilde{\rho}:=\sum_{ k\ge l}\eta_k\delta_{\alpha_k}
\end{equation*}
has also the property that all the polynomials form a dense linear subset of
$L_2(\reals,\widetilde{\rho})$.

 Indeed, on one hand the fact that all the polynomials are in
$L_2(\reals,\sigma)$ implies  the same occurs for
$L_2(\reals,\widetilde{\rho})$. On the other hand, if there is $h\in
L_2(\reals,\widetilde{\rho})$, such that
$\inner{h}{t^m}_{L_2(\reals,\widetilde{\rho})}=0$ for all
$m\in\nats\cup\{0\}$, then
\begin{equation*}
  \sum_{k=l}^\infty\alpha_k^mh(\alpha_k)\eta_k=0\quad\text{for all } m\in\nats\cup\{0\}\,.
\end{equation*}
Thus, by considering the function
\begin{equation*}
  \widetilde{h}(\alpha_k)=
  \begin{cases}
    h(\alpha_k) & k\ge l\,\\
    0 &  k < l\,,
  \end{cases}
\end{equation*}
one obtains that
\begin{equation*}
  \sum_{k=1}^\infty\alpha_k^m\widetilde{h}(\alpha_k)\eta_k=0
\quad\text{for all } m\in\nats\cup\{0\}\,.
\end{equation*}
By the density of the polynomials in $L_2(\reals,\sigma)$, one concludes
that the norm in $L_2(\reals,\sigma)$ of $\widetilde{h}$ vanishes, which in turn implies
that $\norm{h}_{L_2(\reals,\widetilde{\rho})}=0$.

For completing the proof, set
\begin{equation*}
  \widetilde{q}_l:=q_1\,,\quad
  \widetilde{b}_{l-1}^2\widetilde{m}_l^-:=\sum_{k=1}^{l-1}\frac{\eta_k}{\alpha_k-z}\,,
\quad\widetilde{b}_{l}^2\widetilde{m}_l^+:=\sum_{k=l}^\infty\frac{\eta_k}{\alpha_k-z}\,,
\end{equation*}
and substitute these expressions into \eqref{eq:almost-green} to obtain
\begin{equation*}
  -m(z)^{-1}=z-\widetilde{q}_l+\widetilde{b}_{l-1}^2\widetilde{m}_l^-+\widetilde{b}_{l}^2\widetilde{m}_l^+\,.
\end{equation*}
Finally, note that the r.\,h.\,s of the last equation is the $l$-th
Green function of some Jacobi operator by Proposition~\ref{prop:Gkk-formula}.
\end{proof}

\section{Index of determinacy}
\label{sec:stability-index-determinacy}
We begin this section by introducing the following notation. For a
nonnegative Borel measurable function $h$ and a Borel measure $\nu$,
we denote by $h\nu$ the measure which associates to any Borel set $\Delta$
the value
\begin{equation*}
  \int_\Delta hd\nu\,.
\end{equation*}
Thus $h\nu$ is the measure with density $h$ with respect to $\nu$.

The fact that a measure $\mu$ is in $\mathcal{M}$, the set of Borel
measures on $\reals$ with infinite support and all their moments
finite (see Introduction), is indeterminate $N$-extremal or
determinate may be changed by adding or substracting the mass at only
one point.
\begin{proposition}
  \label{prop:duran}
  Let $\mu\in\mathcal{M}$ be indeterminate $N$-extremal.
  \begin{enumerate}[(a)]
  \item If $\lambda\not\in\supp\mu$, then $\mu+a\delta_\lambda$
    ($a>0$) is not $N$-extremal.
  \item If $\lambda\in\supp\mu$, then
    $\mu-\mu(\{\lambda\})\delta_\lambda$ is determinate.
  \end{enumerate}
$\supp\mu$ is the minimal closed set whose complement has $\mu$-zero
measure.
\end{proposition}
\begin{proof}
(a) (Communicated by A. Dur\'an) Let $\widetilde{\mu}$ be an
$N$-extremal measure having the same moments as $\mu$ and such that
\begin{equation}
\label{eq:weight-to-lambda}
\widetilde{\mu}(\{\lambda\})>0\,.
\end{equation}
The existence of such
a $\widetilde{\mu}$ is guaranteed by \cite[Thm.\,3.41]{MR0184042}
and \cite[Thm.\,5]{MR1627806}. Thus, the measures
$\mu+a\delta_\lambda$ and
$\widetilde{\mu}+a\delta_\lambda$ have the same moments, but
\begin{equation*}
  \mu(\{\lambda\})+a<\widetilde{\mu}(\{\lambda\})+a
\end{equation*}
as a consequence of (\ref{eq:weight-to-lambda}) and the fact that
$\mu(\{\lambda\})=0$. The last inequality shows that
$\mu+a\delta_\lambda$ is not $N$-extremal since, by \cite[Thm.\,3.41]{MR0184042}
and \cite[Thm.\,5]{MR1627806}, if an $N$-extremal measure gives weight
to a point, then no other solution of the moment problem can give more
weight to that point.

(b) (\cite[Thm.\,7]{MR638619}) We give an alternative proof based on 
\cite[Thm. 3.4]{MR0184042}. Define
\begin{equation*}
  \widetilde{\mu}:=\mu-\mu(\{\lambda\})\delta_\lambda\,.
\end{equation*}
Note that $\widetilde{\mu}\in\mathcal{M}$ and
$\widetilde{\mu}(\{\lambda\})=0$. If $\widetilde{\mu}$ is
indeterminate, then there exists a solution of the moment problem
$\gamma$ such that $\gamma(\{\lambda\})>0$ due to
\cite[Thm.\,3.41]{MR0184042} (see also \cite[Thm.\,5]{MR1627806}). Now
\begin{equation*}
  \widetilde{\gamma}:=\gamma+ \mu(\{\lambda\})\delta_\lambda
\end{equation*}
is a solution of the moment problem associated with $\mu$
and gives more weight to $\lambda$ than $\mu$ which is a contradiction
\end{proof}
\begin{remark}
  \label{rem:density-polynomials}
  Since the polynomials are dense
  in $L_2(\reals,\mu)$ if and only if  $\mu$ is $N$-extremal, part (a) of
  Proposition~\ref{prop:duran} shows that the density can be destroyed
  by adding just one point mass to the measure.
\end{remark}

\subsection{Characterization of the index of
  determinacy}
\label{sec:defin-char-index}
\begin{definition}
\label{def:index}
For a determinate measure $\mu$, Berg and Dur\'an
introduce in \cite{MR1308001} the index of
determinacy as follows.
\begin{equation*}
  \ind_z\mu=\sup\{k\in\nats\cup\{0\}:\abs{t-z}^{2k}\mu\text{ is determinate}\}\,,
\end{equation*}
where $z\in\complex$.
Since the index of determinacy happens to be constant
\cite[Lem.\,3.5]{MR1308001} at complex numbers outside
the support of $\mu$, one can define
\begin{equation*}
  \ind\mu:=\ind_z\mu\qquad z\not\in\supp\mu\,.
\end{equation*}
\end{definition}

In \cite[Lem.\,2.1]{MR1375156}, the index of determinacy of a measure
is characterized when the measure is multiplied by an arbitrary
polynomial. The next assertion, which follows directly from
results due to C. Berg and A. Dur\'{a}n, describes the general situation.
\begin{proposition}
  \label{prop:berg-duran}
  Let $r$ be a polynomial with simple zeros, $\mu\in\mathcal{M}$ and
\begin{equation*}
  l:=\card\{\text{\rm zeros of }\, r\ \text{\rm outside } \supp\mu\}\,.
\end{equation*}
Then
\begin{enumerate}[(a)]
  \item $\mu$ is determinate and  $\ind\mu=l-1$ if and only if $\abs{r}^2\mu$ is indeterminate and
 $N$-extremal.
 \item $\mu$ is determinate and  $\ind\mu=k\ge l$ if and only if $\abs{r}^2\mu$ is determinate and
  $k=\ind\abs{r}^2\mu +l$.
\item $\mu$ is indeterminate or $\mu$ is determinate and
  $\ind\mu<{l-1}$ if and only if $\abs{r}^2\mu$ is indeterminate and
  not $N$-extremal.
 \end{enumerate}
\end{proposition}
\begin{proof}
  (a) ($\Rightarrow$) Let $a\not\in\supp\mu$ be a zero of $r$. Write
  $r=(t-a)\hat{p}$. Since $\ind\mu=l-1$, we get
  $\ind\abs{\hat{p}}^2\mu=0$ by
  \cite[Lem.\,2.1(ii)]{MR1375156}. Thus, $\abs{t-a}^2\abs{\hat{p}}^2\mu$ is
  indeterminate by Definition~\ref{def:index}. Due to
  \cite[Lem.\,A(1)]{MR1308001} (cf. \cite{mriesz1923}),
  $\abs{t-a}^2\abs{\hat{p}}^2\mu$ is $N$-extremal. ($\Leftarrow$) Now, assume that
  $\abs{r}^2\mu$ is indeterminate $N$-extremal and let $a$ and $\hat{p}$ be
  as before. Using the contrapositive of \cite[Prop.\,3.2]{MR1308001},
  one has $\abs{\hat{p}}^2\mu$ is determinate. $\abs{\hat{p}}^2\mu$
  has zero index of determinacy since, otherwise $\abs{r}^2\mu$ would
  be determinate. Applying again \cite[Lem.\,2.1(ii)]{MR1375156} to
  $\abs{\hat{p}}^2\mu$, one proves the assertion.

  (b) ($\Rightarrow$) This is \cite[Lem.\,2.1(ii)]{MR1375156}.
  ($\Leftarrow$) $\abs{r}^2\mu$ determinate implies $\mu$ is
  determinate by \cite[Prop.\,3.2(i)]{MR1308001}.We must have
  $\ind\mu\ge l$ since $\ind\mu <l$ implies $\abs{r}^2\mu$ is
  indeterminate by \cite[Lem.\,2.1(i)]{MR1375156}.  From
  \cite[Lem.\,2.1(ii)]{MR1375156} follows $k=\ind\mu$

  (c) ($\Rightarrow$) If $\mu$ is indeterminate apply
  \cite[Prop.\,3.2(i)]{MR1308001}. If $\mu$ is determinate then by
  \cite[Lem.\,2.1(i)]{MR1375156} $\abs{r}^2\mu$ is an indeterminate
  measure and by (a) above it cannot be $N$-extremal.  ($\Leftarrow$)
  If $\mu$ is determinate then $\ind \mu<l-1 $ since otherwise we are in
  cases (a) or (b) above.
\end{proof}
\begin{corollary}
  \label{cor:index-weyl-function}
  Let $\rho$ be the spectral measure of a Jacobi operator $J$. For the
  $n$-th Green function $G(z,n)$ of $J$ to be the $l$-th Green function of
  some other Jacobi operator for any $l\in\nats$ it is necessary and sufficient that
  \begin{equation}
    \label{eq:ind-ge}
    \ind\rho\ge \card\{\text{\rm zeros of $\pi_n$ outside $\supp\rho$}\}-1\,.
  \end{equation}
\end{corollary}
\begin{proof}
  Suppose that \eqref{eq:ind-ge} holds. Then, by
  Proposition~\ref{prop:berg-duran}, the polynomials are dense in
  $L_2(\reals,\pi_n^2\rho)$. One direction of the assertion then
  follows from Theorem~\ref{thm:equivalence-green-weyl}. If one
  assumes that
  \begin{equation*}
    \ind\rho<\card\{\text{\rm zeros of $\pi_n$ outside $\supp\rho$}\}-1\,,
  \end{equation*}
  then the polynomials are not dense in $L_2(\reals,\pi_n^2\rho)$ by
  Proposition~\ref{prop:berg-duran} . Therefore $\pi_n^2\rho$ cannot
  be the spectral measure of a Jacobi operator due to
  Theorem~\ref{thm:nec-suf-cond-for-meausure} and then, by
  \eqref{eq:weyl-by-spectral-th} and \eqref{eq:green-integral},
  $G(z,n)$ is not the Weyl $m$-function of a Jacobi operator.
\end{proof}

\begin{lemma}
  \label{lem:det-less-det}
  Let $\mu$ be a determinate measure. If a measure $\nu$ is such that
  $\nu(\mathcal{A})\le\mu(\mathcal{A})$, for any Borel set $\mathcal{A}$,
  then $\nu$ is determinate.
\end{lemma}
\begin{proof}
  (Communicated by C. Berg) Suppose that there is a measure $\sigma$
  different from $\nu$ having the same moments as $\nu$. Then
  $\sigma+\tau$ and $\nu+\tau$ are two measures with the same
  moments. If one takes $\tau=\mu-\nu$, then $\mu=\nu+\tau$ has the
  same moments as $\sigma+\tau$, which is a contradiction because
  $\mu$ is determinate.
\end{proof}

With the help of Definition~\ref{def:index}, one can give more general
and precise statements regarding what happens when mass points are
added or removed from a measure in $\mathcal{M}$. The next statement
is essentially a reformulation of results by C. Berg and A. Dur\'{a}n.
\begin{proposition}
  \label{prop:berg-duran-sum}
 Let $\mathcal{F}\subset\reals$ be a finite set and
\begin{equation}
    \label{eq:beta}
    \beta:=\sum_{\xi\in\mathcal{F}}\beta_\xi\delta_\xi\,,\qquad\beta_\xi>0\,,
  \end{equation}
$\mu\in\mathcal{M}$, and $l:=\card\{\xi\in\mathcal{F}\ \text{\rm outside
}\supp\mu\}$.
\begin{enumerate}[(a)]
\item $\mu$ is determinate and $\ind\mu=l-1$ if and only if
  $\mu+\beta$ is indeterminate $N$-extremal.
\item $\mu$ is determinate and $\ind\mu=k\ge l$ if and only if
  $\mu+\beta$ is determinate and $k=\ind(\mu+\beta)+l$.
\item $\mu$ is indeterminate or $\mu$ is determinate with $\ind\mu<l-1$
  if and only if $\mu+\beta$ is indeterminate and not $N$-extremal.
\end{enumerate}
\end{proposition}
\begin{proof}
  (a)  One direction is \cite[Thm.\,3.6]{MR1308001} and
the converse is \cite[Lem.\,3.7, Thm.\,3.9]{MR1308001}.

  (b) Let

  \begin{equation*}
   \widetilde{\beta}=\beta+\sum_{i=1}^{k+1-l}a_i\delta_{\xi_i},\qquad a_i>0\,
  \end{equation*}
  where $\xi_i\notin\supp(\mu+\beta)$ for $i\in\{1,...,k+1-l \}$. Then
  $\widetilde{\beta}$ is a measure such that $\card\{\text{\rm
    $\xi\in \supp\widetilde{\beta}$ outside } \supp\mu\}=k+1$
  . Applying (a) above we get $\ind \mu=k$ if and only if
  $\mu+\widetilde{\beta}=\mu+\beta+\sum_{i=1}^{k+1-l}a_i\delta_{\xi_i}$
  is indeterminate $N$-extremal and this happens if and only if
  $\ind(\mu + \beta)= k-l$ by (a) again since
  $\xi_i\notin\supp(\mu+\beta)$.

 (c) ($\Rightarrow$) Let $\mathcal{C}\subset\{\xi\in\mathcal{F}\ \text{\rm outside
}\supp\mu\}$ be such that $\card\mathcal{C}=\ind\mu+1<l$. Define
\begin{equation*}
  \widetilde{\gamma}:=\mu+\sum_{\lambda\in\mathcal{C}}\beta_\lambda\delta_\lambda\,.
\end{equation*}
By item (a), $\widetilde{\gamma}$ is indeterminate $N$-extremal. By
Proposition~\ref{prop:duran} and Lemma~\ref{lem:det-less-det},
$\mu+\beta$ is indeterminate not $N$-extremal. ($\Leftarrow$) If $\mu$
is determinate then $\ind\mu<l-1$ since otherwise we are in cases (a)
or (b) above.
\end{proof}
\begin{remark}
  \label{rem:finite-index-discrete}
  A measure of finite index of determinacy is discrete
  (cf. \cite[Cor.\,3.4]{MR1308001}). In view of
  Proposition~\ref{prop:berg-duran-sum}(a), this is a consequence of the fact
  that an indeterminate $N$-extremal measure is discrete
  \cite[Chap.\,3 Sec.\,2 Pag.\,101]{MR0184042}.
\end{remark}

\begin{remark}
  \label{rem:infinite-index-discrete}
  There are measures with infinite index of determinacy being
  discrete. Indeed, take an indeterminate $N$-extremal measure and
  remove the mass at an infinite set of points. By
  Lemma~\ref{lem:det-less-det} and Proposition~\ref{prop:berg-duran-sum}(a), the index of
  determinacy of the modified measure is not finite.
\end{remark}

The following assertion is related to \cite[Rem. p. 231,
Thm.\,5]{MR1343638} (see also \cite{MR1308001} Lemma B and the comment
before Lemma D)
\begin{lemma}
 \label{lem:stability-index-cardinality0}
 Let $\mathcal{I}\subset\reals$ be an infinite discrete set and
 $\widetilde{\mathcal{F}}$ a finite set in $\reals$ such that
 $\mathcal{I}\cap\widetilde{\mathcal{F}}=\emptyset$. Consider a
 sequence 
 $\{\beta_\xi\}_{\xi\in\mathcal{I}\cup\widetilde{\mathcal{F}}}$ of positive
 numbers. Define
 \begin{equation*}
   \mu:=\sum_{\xi\in\mathcal{I}}\beta_\xi\delta_\xi\quad\text{ and
   }\quad
  \widetilde \mu=\mu -\sum_{\xi\in\mathcal{F}}\beta_\xi\delta_\xi
  + \sum_{\xi\in\mathcal{\widetilde F}}\beta_\xi\delta_\xi\,,
 \end{equation*}
 where $\mathcal{F}$ is a finite subset of $\mathcal{I}$. Suppose
 that $\mu\in\mathcal{M}$ is either indeterminate $N$-extremal or
 determinate with finite
 index of determinacy. For
 $\card\mathcal{F}=\card\mathcal{\widetilde F}$ to hold, it is necessary and
 sufficient that  either
$\ind\mu=\ind\widetilde{\mu}$ or $\mu$ and $\widetilde{\mu}$ are
simultaneously indeterminate $N$-extremal.
\end{lemma}
\begin{proof}
($\Rightarrow$)

i) For the case when $\mu$ is indeterminate $N$-extremal, the proof is
essentially given in \cite[Thm.\,8]{MR638619}.

ii) If $0\le\ind\mu=k<+\infty$, choose a set
$\mathcal{A}\subset \reals\setminus
(\mathcal{I}\cup\widetilde{\mathcal{F}})$ such that
$\card\mathcal{A}=k+1$ and consider the measure
$$\mu +\sum_{\xi\in\mathcal{A}} a_\xi\delta_\xi\,, $$
where $ a_\xi>0$. By Proposition~\ref {prop:berg-duran-sum} (a), this
measure is indeterminate $N$-extremal. It then follows from i) that the
measure
$$\widetilde \mu +\sum_{\xi\in\mathcal{A}} a_\xi\delta_\xi$$
is indeterminate $N$-extremal too. Using again
Proposition~\ref{prop:berg-duran-sum} (a) we get that
$\ind \widetilde \mu=k=\ind \mu$.

($\Leftarrow$)

Assume without loss of generality that $\card\mathcal{F}<\card\mathcal{\widetilde F}$
and let $\mathcal{G}\subset\mathcal{\widetilde F}$ be such that
$\card\mathcal{F}=\card\mathcal{G}$. Then
$$\widetilde \mu=\nu+\sum_{\xi\in {\mathcal{\widetilde F}\setminus
   \mathcal{G}}}\beta_\xi\delta_\xi\,,$$
where
$$\nu=\mu-\sum_{\xi\in\mathcal{F}}\beta_\xi\delta_\xi + \sum_{\xi\in
  \mathcal{G}}\beta_\xi\delta_\xi\,. $$ By what was proven in i) and ii)
above, either $\ind\nu=\ind\mu$ or $\mu$ and $\nu$ are simultaneously
indeterminate $N$-extremal.  By Proposition~\ref{prop:berg-duran-sum}
neither $\ind\widetilde \mu=\ind\mu$ nor $\mu$ and $\widetilde{\mu}$ are simultaneously
indeterminate $N$-extremal since
$\mathcal{\widetilde F}\setminus\mathcal{G}$ is not in the support of $\nu$.
\end{proof}
A consequence of the previous lemma is the following result,
\begin{lemma}
  \label{lem:stability-index-cardinality}
  Let $\mathcal{I}\subset\reals$ be an infinite discrete set and
  $\{\beta_\xi\}_{\xi\in\mathcal{I}}$ be a sequence of positive
  numbers. Assume that $\mathcal{F}_1,\mathcal{F}_2\subset\mathcal{I}$
  are finite sets and
  $\sum_{\xi\in\mathcal{I}\setminus\mathcal{F}_1}\beta_\xi\delta_\xi$
  is $N$-extremal not having infinite index of determinacy.
  \begin{equation*}
    \card\mathcal{F}_1=\card\mathcal{F}_2
  \end{equation*}
if and only if either
\begin{equation*}
  \ind\sum_{\xi\in\mathcal{I}\setminus\mathcal{F}_1}\beta_\xi\delta_\xi=
\ind\sum_{\xi\in\mathcal{I}\setminus\mathcal{F}_2}\beta_\xi\delta_\xi
\end{equation*}
or the measures
$\sum_{\xi\in\mathcal{I}\setminus\mathcal{F}_1}\beta_\xi\delta_\xi$
and
$\sum_{\xi\in\mathcal{I}\setminus\mathcal{F}_2}\beta_\xi\delta_\xi$
are simultaneously indeterminate $N$-extremal.

\end{lemma}
\begin{proof}
 Observe that
\begin{equation*}
\sum_{\xi\in\mathcal{I}\setminus\mathcal{F}_1}\beta_\xi\delta_\xi=
\sum_{\xi\in\mathcal{I}\setminus\mathcal{F}_2}\beta_\xi\delta_\xi
\ -\!\!\!\!\!\!\sum_{\xi\in(\mathcal{I}\setminus\mathcal{F}_2)\cap F_1}\beta_\xi\delta_\xi
\ +\!\!\!\!\!\!\sum_{\xi\in(\mathcal{I}\setminus\mathcal{F}_1)\cap F_2}\beta_\xi\delta_\xi
\end{equation*}
and apply lemma \ref {lem:stability-index-cardinality0}, noting that
$\card(\mathcal{I}\setminus\mathcal{F}_2)\cap
F_1=\card(\mathcal{I}\setminus\mathcal{F}_1)\cap F_2$ if and only if
$\card F_1= \card F_2$.

\end{proof}
\begin{theorem}
  \label{thm:other-criterion}
  Let $J$ be a Jacobi operator (see
  Definition~\ref{def:jacobi-operator}) and $\rho$ its
  spectral
  measure. Assume that $\rho$ is a determinate measure and $r$ is a polynomial with
  simple zeros. Then
  \begin{equation}
    \label{eq:condition-other-criterion}
    \ind\rho\ge\card(\{\text{zeros of
  r}\}\setminus \sigma(J))-1
  \end{equation}
if and only if $r(J)e_1$ is a cyclic vector for $J$.
\end{theorem}
\begin{proof}($\Leftarrow$)
  Let $u=r(J)e_1$ and assume that $u$ is a cyclic vector for $J$, i.\,e.,
  \begin{equation}
    \label{eq:density-powers}
    \cc{\Span_{k\in\nats\cup\{0\}}\{J^ku\}}=\l_2(\nats)\,.
  \end{equation}
  Since $u$ is a cyclic vector, it follows from \cite[Sec. 69,
  Thm. 2]{MR1255973} that there is a unitary map
  $\Phi:L_2(\reals,\mu)\to l_2(\nats)$, where
  $\mu(\Delta):=\inner{u}{E(\Delta)u}$ for any Borel set $\Delta$ of
  $\reals$ (see Section~\ref{sec:jacobi_operators}), such that
  $\Phi J\Phi^{-1}$ is the operator of multiplication by the
  independent variable. Thus, \eqref{eq:density-powers} is equivalent
  to
\begin{equation*}
  \cc{\Span_{k\in\nats\cup\{0\}}\{t^k\}}=L_2(\reals,\mu)\,.
\end{equation*}
For finishing this part of the proof, it only remains to note that
\begin{equation*}
  \mu=\abs{r}^2\rho
\end{equation*}
and recur to Proposition~\ref{prop:berg-duran} recalling that
$N$-extremality is equivalent to density of polynomials (see
Introduction) and that $\sigma(J)=\supp\rho$.

($\Rightarrow$) First note that $r(J)e_1$ is in
$\dom(J^k)$ for all $k\in\nats\cup\{0\}$ since $r(J)e_1$ is a finite
vector, that is, the corresponding sequence has a finite number of
elements different from zero. By Proposition~\ref{prop:berg-duran},
\eqref{eq:condition-other-criterion} implies
\begin{equation}
 \label{eq:density-r-squared}
  \cc{\Span_{k\in\nats\cup\{0\}}\{t^k\}}=L_2(\reals,\abs{r}^2\rho)
\end{equation}
and let $w\in l_2(\nats)$ be such that
\begin{equation*}
  \inner{J^kr(J)e_1}{w}=0\quad\text{ for all }k=0,1,2,\dots
\end{equation*}
This means that
\begin{equation}
  \label{eq:cyclic-cdense-aux}
  \int_\reals h(t)t^k\cc{r(t)}d\rho=0\quad\text{ for all }k=0,1,2,\dots,
\end{equation}
where $w=\Phi h$ (see Section~\ref{sec:jacobi_operators}). If one
writes $h=\widetilde{h}r$, then
$\widetilde{h}\in L_2(\reals,\abs{r}^2\rho)$ since
\begin{equation*}
 +\infty> \int_\reals \abs{h}^2d\rho=
\int_\reals \abs{\widetilde{h}}^2\abs{r}^2d\rho\,.
\end{equation*}
Hence, taking into account (\ref{eq:cyclic-cdense-aux}), one has, for
any $k\in\nats\cup\{0\}$,
\begin{equation*}
  0=\int_\reals h(t)t^k\cc{r(t)}d\rho=\int_\reals
  \widetilde{h}(t)t^k\abs{r(t)}^2d\rho
   =\inner{t^k}{\widetilde{h}}_{L_2(\reals,\abs {r}^2\rho)}\,.
\end{equation*}
Due to (\ref{eq:density-r-squared}), this implies that
$\widetilde{h}=0$, viz.,
\begin{equation*}
 0=\norm{\widetilde{h}}_{L_2(\reals,\abs{r}^2\rho)}^2=
\int_\reals \abs{\widetilde{h}}^2\abs{r}^2d\rho=\int_\reals \abs{h}^2d\rho\,.
\end{equation*}
Whence $\norm{h}_{L_2(\reals,\rho)}=0$. Thus, the vector $w$ must
vanish which means that $r(J)e_1$ is a cyclic vector.
\end{proof}

In fact, as shown below, $\ind\rho$ is the only natural number
satisfying the assertion of Theorem~\ref{thm:other-criterion}.
\begin{corollary}
  \label{cor:other-criterion}
  Let $J$ and $\rho$ and $r$ be as in
  Theorem~\ref{thm:other-criterion}. If
 $k\in\nats\cup\{0\}$ is such that $r(J)e_1$ is a cyclic
vector for $J$ whenever
\begin{equation}
  \label{eq:cyclic-one-option}
  \card(\{\text{zeros of $r$}\}\setminus \sigma(J))\le k+1
\end{equation}
and it is not a cyclic vector for $J$ whenever
\begin{equation}
  \label{eq:cyclic-two-option}
 \card(\{\text{zeros of $r$}\}\setminus \sigma(J))> k+1\,,
\end{equation}
then $k=\ind\rho$.
\end{corollary}
\begin{proof}
Suppose that $\ind\rho<k$. Choose a polynomial $r$ such that
$k=\card(\{\text{zeros of $r$}\}\setminus \sigma(J))-1$. It follows
from \eqref{eq:cyclic-one-option} that $r(J)e_1$ is a cyclic
vector. But $\ind\rho<\card(\{\text{zeros of $r$}\}\setminus
\sigma(J))-1$ implies that $r(J)e_1$ is not cyclic by
Theorem~\ref{thm:other-criterion}. So, assuming $\ind\rho<k$ leads to
a contradiction. Therefore $\ind\rho\ge k$. Let $\ind\rho > k$.
If $r$ is such that $\ind\rho=\card(\{\text{zeros of $r$}\}\setminus \sigma(J))-1$,
then Theorem~\ref{thm:other-criterion} implies that $r(J)e_1$
is cyclic vector. But in this case $k<\card(\{\text{zeros of
  $r$}\}\setminus \sigma(J))-1$ and \eqref{eq:cyclic-two-option}
implies that $r(J)e_1$ is a not cyclic vector. We get again a
contradiction. Therefore $k=\ind\rho$.
\end{proof}
\subsection{Stability of the index of determinacy}
Let us study the stability of the index of determinacy and the
$N$-extremality for measures. First we deal with the case when the
support of the measure does not change.
\begin{proposition}
  \label{prop:changing-weights}
  Changing the weights of a measure can change its index of
  determinacy.
\end{proposition}
\begin{proof}
  Consider the following criterion for a measure to be determinate
  \cite[Thm.\,5.2,\,pag.\,84]{MR0481888}: If there is $\epsilon >0$
  such that
  \begin{equation}
    \label{eq:criteria-determinacy}
    \int_\reals e^{\epsilon\abs{t}}d\mu<\infty\,,
  \end{equation}
  then $\mu$ is determinate. Thus, an indeterminate, $N$-extremal measure $\nu$,
  can be transformed into $\mu$ by changing the weights so that
  \eqref{eq:criteria-determinacy} holds. Now consider a measure
  $\sigma$ of index $n$ obtained from $\nu$ by removing the mass at
  $n+1$ points. The measure $\widetilde{\sigma}$ obtained by removing
  from $\mu$ the mass at the same $n+1$ points has index of
  determinacy greater than $n$.
\end{proof}
\begin{proposition}
  \label{prop:finite-change-weights}
  By changing a finite number of weights the index of determinacy is
  preserved.
\end{proposition}
\begin{proof}
  From Proposition~\ref{prop:berg-duran-sum}(a), a measure has an infinite
  index of determinacy if and only if, after adding any finite number
  of mass points, it remains determinate. Thus, changing a finite
  number of weights do not alter the infinite index of
  determinacy. Suppose that $\ind\mu<\infty$, then, by
  Proposition~\ref{prop:berg-duran-sum}(a) (see also
  \cite[Pag.\,129]{MR1375156}), $\mu$ is obtained by removing from an
  indeterminate $N$-extremal measure $\mu_0$ the mass at a finite set
  of points. According to \cite[Thm.\,5(b)]{MR1343638} the measure
  $\widetilde{\mu}_0$ obtained by modifying the weight of $\mu_0$ at
  one mass point is indeterminate $N$-extremal. Adding to
  $\widetilde{\mu}_0$ the same masses at the same points that were
  substracted from $\mu_0$ to obtain $\mu$ yields a measure $\eta$ with the
  same index of determinacy as $\mu$. Note that $\eta$ is equal to
  $\mu$ with one weight modified.
\end{proof}
\begin{theorem}
  \label{thm:main-minimal-equal-indices}
  Let $J$ and $\widehat{J}$ be Jacobi operators as given in
  Definition~\ref{def:jacobi-operator} with spectral measures $\rho$
  and $\widehat{\rho}$, respectively. Suppose that, for some
  $n\in\nats\setminus\{1\}$,
  \begin{equation}
    \label{eq:equal-number-of-ceros-in-supp}
    \card(\sigma(J_n^-)\cap\sigma(J))=
    \card(\sigma(\widehat{J}_n^-)\cap\sigma(\widehat{J}))\,,
  \end{equation}
where $J_n^-$ and $\widehat{J}_n^-$ are given in Definition~\ref{def:submatrices}.
Consider the measure $\rho_n$ given in \eqref{eq:def-rho-n} and the
corresponding measure $\widehat{\rho}_n$ for $\widehat{J}$.
If $\ind\rho_n=\ind\widehat{\rho}_n$ or $\rho_n$ and
$\widehat{\rho}_n$ are simultaneously indeterminate $N$-extremal, then
$\ind\rho=\ind\widehat{\rho}$. Conversely, if
\begin{equation*}
  \ind\rho=\ind\widehat{\rho}\ge n-(\card(\sigma(J_n^-)\cap\sigma(J))+1)\,,
\end{equation*}
then, $\ind\rho_n=\ind\widehat{\rho}_n$ or $\rho_n$ and
$\widehat{\rho}_n$ are simultaneously indeterminate $N$-extremal.
\end{theorem}
\begin{proof}
  Due to the fact that $\sigma(J_n^-)$ is simple,
  \eqref{eq:equal-number-of-ceros-in-supp} and
  Remark~\ref{rem:zeros-poly} imply that the number of zeros of
  $\pi_n$ that are not in the $\supp\rho$ is equal to the number of
  zeros of $\widehat{\pi}_n$ that are not in the
  $\supp\widehat{\rho}$. The assertion then follows from Proposition
  \ref{prop:berg-duran}
\end{proof}
  The following statement gives a criterion for two measures with the same support and
  different weights to have the same index of determinacy.
\begin{corollary}
  \label{cor:isospectral-matrices-equal-indices}
  Let $J$, $\widehat{J}$, $\rho$, $\widehat{\rho}$, $\rho_n$, and
  $\widehat{\rho}_n$ be as in the previous theorem.
  Assume that $J$ and $\widehat{J}$ are isospectral Jacobi operators. If, for
  some $n\in\nats$, $\rho_n=\widehat{\rho}_n$ and $\rho_n$ is
  $N$-extremal, then $\ind\rho=\ind\widehat{\rho}$
\end{corollary}
\begin{proof}
 For any Borel set $\Delta$,
  \begin{equation}
    \label{eq:equal-measures}
    \int_\Delta\pi_n^2(t)d\rho(t)=\rho_n(\Delta)=\widehat{\rho}_n(\Delta)=
\int_\Delta\widehat{\pi}_n^2(t)d\widehat{\rho}(t)\,.
  \end{equation}
  This implies that the zeros of $\pi_n$ that are in the $\supp\rho$
  coincide with the zeros of $\widehat{\pi}_n$ that are in the
  $\supp\widehat{\rho}$. Therefore
  \eqref{eq:equal-number-of-ceros-in-supp} holds. It remains to apply
  Theorem~\ref{thm:main-minimal-equal-indices}.
\end{proof}

If $\rho\in\mathcal{M}$ is determinate, then $\rho$ is the spectral
measure of a Jacobi operator $J$ in the sense of
\eqref{eq:spectral-function-def}. In this case, $J$ is the unique
self-adjoint extension of $J_0$ (see Definition~\ref{def:j-nought}), i.\,e., $J_0$ is essentially
self-adjoint \cite[Thm.\,2.2]{MR0184042}, \cite[Thm.\,2]{MR1627806}.

The following assertion appears in \cite[Thm.\,1]{MR1353377}. We
reproduce it here with a brief proof for the reader's convenience.

\begin{proposition}
  \label{prop:berg-duran-operator-powers}
  Let $\rho$ be a determinate measure. For the measure $\rho$ to have
  index of determinacy $k$, it is necessary and sufficient that
  $J_0^l$ is essentially self-adjoint for $l=1,\dots,k+1$ and
  $J_0^{k+2}$ is not essentially self-adjoint. The measure $\rho$ has infinite index
  of determinacy if and only if $J_0^l$ is essentially self-adjoint
  for all $l\in\nats$.
\end{proposition}
\begin{proof}
  Let $\mathbb{P}$ be the set of all polynomials, i.\,e.,
  \begin{equation*}
    \mathbb{P}=\left\{\sum_{k=0}^Na_kt^k:\ N\in\nats\cup\{0\},
   \ t\in\reals,\ a_k\in\complex\right\}\,.
  \end{equation*}
  The unitary map
  $\Phi$ introduced in Section~\ref{sec:jacobi_operators} satisfies
  \eqref{eq:unitary-map} and therefore
  \begin{equation*}
    \Phi\mathbb{P}=\dom(J_0)\,.
  \end{equation*}
  $J_0^l$ is essentially self-adjoint if and only if
  $\cc{\ran(J_0^l\pm iI)}=l_2(\nats)$ \cite[Cor. to Thm
  VIII.3]{MR751959}. This implies, by means of the unitary map $\Phi$,
  that this happens if and only if
  \begin{equation}
    \label{eq:density-aux}
    \cc{(t^l\pm
    i)\mathbb{P}}=L_2(\reals,\rho)\,.
  \end{equation}
By \cite[Lemma]{MR1353377}, \eqref{eq:density-aux} is equivalent to
\begin{equation}
\label{eq:density-aux-2}
  \cc{\mathbb{P}}=L_2(\reals,(1+t^{2l})\rho)\,.
\end{equation}
It follows from
\begin{equation*}
  1\le\frac{(1+x^2)^l}{1+x^{2l}}\le 2^{l-1}
\end{equation*}
that the polynomials are dense in $L_2(\reals,(1+t^{2l})\rho)$ if and
only if they are dense in $L_2(\reals,(1+t^{2})^l\rho)$. Thus, by
Definition~\ref{def:index}, $\ind\rho=k$ if and only if
\eqref{eq:density-aux-2} is satisfied for $l=1,\dots,k+1$ but does not
hold for $l=k+2$.
\end{proof}

As a consequence of the previous proposition, one has the following
assertion.
\begin{corollary}
  \label{cor:rho-sigma-n}
  Let $J$ be a Jacobi operator and $\rho$ its spectral measure.  If
  $\rho$ is determinate, then the index of determinacy of $\rho$
  coincides with the index of determinacy of $\sigma_n$, as defined in
  (\ref {eq:def-sigma-n}), for any $n\in\nats$. The measure $\rho$ is
  indeterminate $N$-extremal if and only if $\sigma_n$ is
  indeterminate $N$-extremal.
\end{corollary}
\begin{proof}
  Define $B:=\mathbb{O}\oplus J_n^+\eval{\dom(J_0)}$, where
  $\mathbb{O}$ is the null operator in $\Span\{e_k\}_{k=1}^n$ and
  $\oplus$ indicates that we are considering the orthogonal sum of
  operators (see \cite[Sec.\,3.6]{MR1192782}).  Note that the domain
  of $J_0^l$ and $B^l$ is the same for all $l\in\nats$. Since the
  matrix corresponding to the operator $J_0$ is tridiagonal, there
  exists a finite rank operator $C$ such that $B^l+C=J_0^l$ for any
  $l\in\nats$. Note that the rank of $C$ depends on $l$ and $n$. By
  the Kato-Rellich theorem (see \cite[Chap.\,5, Sec.\,4,
  Thm.\,4.4]{MR0407617}) $B^l$ is essentially self-adjoint if and only
  if $J_0^l$ is essentially self-adjoint. Now, since
  $B^l=\mathbb{O}^l\oplus (J_n^+\eval{\dom(J_0)})^l$, $B^l$ is
  essentially self-adjoint if and only if $(J_n^+\eval{\dom(J_0))^l}$
  is essentially self-adjoint and the result follows from Proposition
  \ref {prop:berg-duran-operator-powers}.
\end{proof}
\begin{corollary}
  \label{cor:finite-rank-pert-index-determinacy}
  Let $J$ and $\widetilde{J}$ be Jacobi operators as defined in
  Section~\ref{sec:jacobi_operators} such that $\widetilde{J}= J+C$,
  where $\rank(C)<\infty$, and denote by $\rho$ and $\widetilde{\rho}$
  the corresponding spectral measures. If $\rho$ is determinate, then
  $\widetilde{\rho}$ is determinate and
  $\ind\rho=\ind\widetilde{\rho}$. If $\rho$ is indeterminate
  $N$-extremal, then $\widetilde{\rho}$ is indeterminate $N$-extremal.
\end{corollary}
\begin{proof}
  Since $\rank(C)<\infty$, there is $n\in\nats$ such that
  $J_n^+=\widetilde{J}_n^+$. Therefore, taking into account that,
  according to \cite[Addenda and Problems to Chap.\,1]{MR0184042}, the
  matrix representations of $J$ and $J_n^+$ are simultaneously either
  limit circle case or limit point case, one concludes that $\rho$ and
  $\widetilde{\rho}$ are simultaneously either determinate or
  indeterminate $N$-extremal. If $\rho$ is determinate, then the
  assertion follows from Corollary~\ref{cor:rho-sigma-n}.
\end{proof}

\begin{remark}
  \label{rem:example}
  In the previous proof, one could have used \cite[Chap.\,5, Sec.\,4,
  Thms.\,4.3 and 4.4]{MR0407617} (Kato-Rellich theorem) to show that
  $J_0$ and $\widetilde{J}_0$ are simultaneously either essentially
  self-adjoint or not.
\end{remark}

The next assertion uses the measures introduced in (\ref
{eq:def-sigma-n}), (\ref{eq:def-mu-n}) and (\ref{eq:def-rho-n}).
\begin{theorem}
  \label{thm:oaxaca}
  For the measure $\rho_n$ to be determinate with index of determinacy
  $k$ (indeterminate
  $N$-extremal) it is necessary and sufficient that $\sigma_n+\mu_n$ is
  determinate with index of determinacy $k$ (indeterminate $N$-extremal).
\end{theorem}
\begin{proof}
  Due to Proposition~\ref{prop:berg-duran}, the $\ind\rho_n=k$ if and
  only if $$\ind\rho=k+\card(\sigma(J_n^-)\setminus\sigma(J))\,.$$
  This is so, because the set of zeros of $\pi_n$ is the spectrum of
  $J_n^-$ (see Remark~\ref{rem:zeros-poly}).  Since, according to
  Corollary~\ref{cor:rho-sigma-n}, $\ind\rho=\ind\sigma_n$ one has,
  using Proposition~\ref{prop:berg-duran-sum},
  \begin{equation*}
   \ind(\sigma_n +\mu_n)=k+\card(\sigma(J_n^-)\setminus\sigma(J))-
   \card(\sigma(J_n^-)\setminus\sigma(J_n^+))\,.
  \end{equation*}
  In view of Remark~\ref{rem:intersections-spectra} and
  Proposition~\ref{prop:berg-duran-sum}(b), the last expression yields that
  $\ind(\sigma_n +\mu_n)=k$ if and only if $\ind\rho_n=k$.

  By Proposition~\ref{prop:berg-duran}(a), the measure $\rho_n$ is
  indeterminate $N$-extremal if and only if
  \begin{equation*}
    \ind\rho=\card\{\text{zeros of $\pi_n$ outside }\supp\rho\}-1\,.
  \end{equation*}
Using Remark~\ref{rem:zeros-poly} and Corollary~\ref{cor:rho-sigma-n},
one concludes that the last expression is equivalent to
\begin{equation*}
  \ind\sigma_n= \card(\sigma(J_n^-)\setminus\sigma(J))-1=
\card(\sigma(J_n^-)\setminus\sigma(J_n^+))-1\,.
\end{equation*}
This happens if and only if that $\sigma_n+\mu_n$ is indeterminate $N$-extremal by
Proposition~\ref{prop:berg-duran-sum}(a).
\end{proof}
\begin{corollary}
  \label{cor:last-one}
  Let $\rho$ be the spectral measure of some Jacobi operator $J$ as in
  \eqref{eq:spectral-function-def}. Define the measure $\beta$ by
  \eqref{eq:beta} with
  \begin{equation*}
    \card(\mathcal{F}\setminus\supp\rho)=\card(\sigma(J_n^-)\setminus\sigma(J))\,.
  \end{equation*}
 The measure $\rho_n$ has index of determinacy $k$ (is indeterminate
 $N$-extremal) if and only if $\rho+\beta$ has index of determinacy $k$ (is indeterminate
 $N$-extremal).
\end{corollary}
\begin{proof}
  The assertion follows from Theorem~\ref{thm:oaxaca}, taking into
  account Corollary~\ref{cor:rho-sigma-n} and using the same reasoning
  as in the proof of Proposition~\ref{prop:berg-duran-sum}(b).
\end{proof}
\section{Inverse spectral problems}
  \label{subsec:sol-inv-problem}
Let $J$ be the Jacobi operator associated with the matrix
\eqref{eq:jm-0} as in Section~\ref{sec:jacobi_operators}. Fix
$n\in\nats\setminus\{1\}$ and consider, along with the self-adjoint operator $J$,
the operator
\begin{equation}
\label{eq:def-tilde-j}
\begin{split}
  \widetilde{J}(n)=J &+
  [q_n(\theta^2-1)+\theta^2h]\inner{e_n}{\cdot}e_n \\
  &+  b_n(\theta-1)(\inner{e_n}{\cdot}e_{n+1} +
  \inner{e_{n+1}}{\cdot}e_n) \\
  &+ b_{n-1}(\theta-1)(\inner{e_{n-1}}{\cdot}e_{n} +
  \inner{e_{n}}{\cdot}e_{n-1})
  \,,\quad \theta>0\,,
  \quad h\in\mathbb{R}\,,
\end{split}
\end{equation}
where it has been assumed that $b_0=0$. Clearly, $\widetilde{J}(n)$ is a
self-adjoint extension of the operator whose matrix representation
with respect to the canonical basis in $l_2(\mathbb{N})$ is a Jacobi
matrix obtained from (\ref{eq:jm-0}) by modifying the entries
$b_{n-1},q_n,b_n$. For instance, if $n>2$, $\widetilde{J}(n)$ is a
self-adjoint extension (possibly not proper) of the operator whose
matrix representation is
\begin{equation}
  \label{eq:jm-1}
  \begin{pmatrix}
q_1 & b_1 & 0 & 0 & 0 & 0 & \cdots \\[1mm]
b_1 & \ddots & \ddots & 0 & 0 & 0 & \cdots \\[1mm]
0  &  \ddots  & q_{n-1} & \theta b_{n-1} & 0 & 0 & \cdots\\
0 & 0 & \theta b_{n-1} & \theta^2(q_n+h) & \theta b_n & 0 & \cdots\\
0 & 0 & 0 & \theta b_{n} & q_{n+1} &  b_{n+1} & \\
0 & 0 & 0 & 0 & b_{n+1} & q_{n+2} & \ddots\\
\vdots & \vdots & \vdots & \vdots & & \ddots & \ddots
  \end{pmatrix}\,.
\end{equation}
Note that $\widetilde{J}(n)$ is obtained from $J$ by a rank-three
perturbation when $n>1$, and a rank-two perturbation otherwise.

Define
\begin{equation}
  \label{eq:gamma-hold}
  \gamma:=\frac{\theta^2h}{1-\theta^2}\,.
\end{equation}
Consider the following inverse problem:

Given two sequences $S$ and $\widetilde{S}$ without finite points of
accumulation, $n\in\nats\setminus\{1\}$ and $\gamma\in\reals\setminus S$, find a
Jacobi operator $J$ and parameters $\theta$ and $h$ such that
$\sigma(J)=S$ and $\sigma(\widetilde{J}(n))=\widetilde{S}$ and
\eqref{eq:gamma-hold} holds. We denote
this inverse spectral problem by $(S,\widetilde{S},n,\gamma)$. The
operator $J$ is called a solution of the inverse problem $(S,\widetilde{S},n,\gamma)$.

When $n>1$, it was shown in \cite[Thms.\,5.6]{MR3634444} that
if there is a solution, then there is an infinite set of
solutions. Necessary and sufficient conditions on $S$ and
$\widetilde{S}$ for the existence of solutions of the inverse problem
are given in \cite[Thms.\,5.9]{MR3634444}.

\begin{remark}
\label{rem:same-green}
All solutions of this inverse spectral problem have the same
Green function at $n$ \cite[Prop.\,5.3]{MR3634444} given by
\begin{equation*}
  G(z,n)=\frac{\mathfrak{M}_n(z)-\theta^2}{(1-\theta^2)(\gamma-z)}
\end{equation*}
(see \cite[Eq.\,4.2]{MR3634444}), where the function
$\mathfrak{M}_n$ is univocally determined by the sequences $S$ and
$\widetilde{S}$ \cite[Prop.\,4.13]{MR3634444}. Moreover,
$\gamma$ and $\mathfrak{M}_n$ uniquely determine $\theta$ (see proof
of \cite[Prop.\,5.4]{MR3634444}).
\end{remark}
\begin{theorem}
  \label{thm:stability-determ-index-weights}
  Let $J$ and $\widehat{J}$ be Jacobi operators which solve the
  inverse problem $(S,\widetilde{S},n,\gamma)$ and $\rho$ and $\widehat{\rho}$ be the
  corresponding spectral measures. Then, either
  \begin{equation*}
    \ind\rho=\ind\widehat{\rho}
  \end{equation*}
or $\rho$ and $\widehat{\rho}$ are simultaneously indeterminate $N$-extremal.
\end{theorem}
\begin{proof}
  Due to Remark~\ref{rem:same-green}, $J$ and $\widehat{J}$ have the
  same function $G(z,n)$. According to
  \cite[Prop.\,3.5]{MR3634444}, one writes
  \begin{equation*}
    -G(z,n)^{-1}=z-q_n+\sum_{\alpha\in\mathcal{I}}\frac{\eta_\alpha}{\alpha-z}\,,
  \end{equation*}
where $\mathcal{I}$ is a discrete subset of $\reals$. By
Proposition~\ref{prop:Gkk-formula}
\begin{equation*}
 \sum_{\alpha\in\mathcal{I}}\frac{\eta_\alpha}{\alpha-z}=
\begin{cases}
  b_n^2m_n^+(z)+b_{n-1}^2m_n^-(z)& \\
  b_n^2\widehat{m}_n^+(z)+b_{n-1}^2\widehat{m}_n^-(z)\,,&
\end{cases}
\end{equation*}
where $m_n^\pm$ are given in Definition~\ref{def:submatrices} and
$\widehat{m}_n^\pm$ are the corresponding functions for
$\widehat{J}$. Thus, using the notation introduced in
\eqref{eq:def-sigma-n}, one has
\begin{equation*}
   b_n^2\sigma_n=\sum_{\alpha\in\mathcal{I}\setminus\mathcal{F}}\eta_\alpha\delta_\alpha\qquad
 b_n^2\widehat{\sigma}_n=\sum_{\alpha\in\mathcal{I}\setminus\widehat{\mathcal{F}}}\eta_\alpha\delta_\alpha
\end{equation*}
where $\widehat{\sigma}_n$ is defined as $\sigma_n$ for the function
$\widehat{m}^+_n$  and  $\card\mathcal{F}=\card\widehat{\mathcal{F}}=n$.
By Lemma~\ref{lem:stability-index-cardinality}, either
\begin{equation*}
  \ind\sigma_n=\ind\widehat{\sigma}_n
\end{equation*}
or $\sigma_n$ and $\widehat{\sigma}_n$ are simultaneously
indeterminate $N$-extremal. Thus, Corollary~\ref{cor:rho-sigma-n}
completes the proof.
\end{proof}
\begin{theorem}
  \label{thm:place-of-perturbation}
  Let $J$ and $J'$ be solutions of the inverse problems $(S,\widetilde{S},n,\gamma)$ and
  $(S,\widetilde{S},n',\gamma)$ respectively. Denote by $\rho$ and
  $\rho'$ the spectral measures corresponding to $J$ and $J'$ and
  assume that $\ind\rho<+\infty$. Either
  \begin{equation*}
    \ind\rho=\ind\rho'
  \end{equation*}
or $\rho$ and
  $\rho'$ are simultaneously indeterminate $N$-extremal if and only if
  \begin{equation*}
    n=n'\,.
  \end{equation*}
\end{theorem}
\begin{proof}
  ($\Leftarrow$) This is
  Theorem~\ref{thm:stability-determ-index-weights}.

  ($\Rightarrow$) Since $S$ and $\widetilde{S}$ univocally determine
  $\mathfrak{M}_n$ (see Remark~\ref{rem:same-green}), one has
  \begin{equation*}
    \mathfrak{M}_n(z)=\mathfrak{M}_{n'}(z)\qquad\text{ for all }
z\in\complex\setminus(S\setminus\widetilde{S})\,.
  \end{equation*}
Again, Remark~\ref{rem:same-green} yields
\begin{equation*}
  G(z,n)=G'(z,n')\,,
\end{equation*}
where $G'(z,n')$ is the $n'$-th Green function of $J'$. Repeating the
argumentation of the previous theorem's proof, one arrives at
\begin{equation}
 \label{eq:sigma-sigma-prime}
     b_n^2\sigma_n=\sum_{\alpha\in\mathcal{I}\setminus\mathcal{F}}\eta_\alpha\delta_\alpha\qquad
 (b'_{n'})^2\sigma'_{n'}=\sum_{\alpha\in\mathcal{I}\setminus\mathcal{F}'}\eta_\alpha\delta_\alpha\,,
\end{equation}
where
\begin{equation}
  \label{eq:f-f-prime}
  n=\card\mathcal{F}\quad\text{ and }\quad n'=\card\mathcal{F}'
\end{equation}
The hypothesis and Corollary~\ref{cor:rho-sigma-n} imply that either
\begin{equation*}
  \ind\sigma_n=\ind\sigma'_{n'}
\end{equation*}
or $\sigma_n$ and $\sigma'_{n'}$ are simultaneously indeterminate
$N$-extremal. To conclude the proof, one applies
Lemma~\ref{lem:stability-index-cardinality} to
\eqref{eq:sigma-sigma-prime} and \eqref{eq:f-f-prime}.
\end{proof}
\begin{remark}
  \label{rem:infinite-finite-index}
  Under the assumption that $S,\widetilde{S},\gamma$ are fixed, if
  $\rho$ in the previous theorem is such that $\ind\rho<+\infty$, then
  the place of the perturbation $n$ is determined uniquely by
  $\ind\rho$. If $\ind\rho=\infty$, then there are several possible
  values of $n$. This happens, in particular, to the inverse problem
  for finite Jacobi matrices.
\end{remark}
\begin{remark}
  \label{rem:other-case-gamma}
  The inverse spectral problem for which $\gamma\in S$ is treated
  analogously.
\end{remark}
\subsection*{Acknowledgments}
We thank C. Berg, A. Dur\'an, and F. Marcell\'an for valuable remarks.

\def\cprime{$'$} \def\lfhook#1{\setbox0=\hbox{#1}{\ooalign{\hidewidth
  \lower1.5ex\hbox{'}\hidewidth\crcr\unhbox0}}} \def\cprime{$'$}

\end{document}